\documentclass
{llncs}
\usepackage[utf8]{inputenc}

\usepackage[color=yellow]{todonotes}
\usepackage{graphicx}
\usepackage{xcolor}
\usepackage{amsmath,amssymb}
\usepackage{lineno}
\usepackage{wrapfig}
\usepackage[symbol]{footmisc}

\spnewtheorem{Remark}{Remark}{\bfseries}{\itshape}

\newcommand{\canon}[1]{\mathcal{C}(#1)}

\newcommand{\slope}[1]{\phi(#1)}

\newcommand{\innerCone}[1]{$C^{\mathrm{in}}_{#1}$}
\newcommand{\outerCone}[1]{$C^{\mathrm{out}}_{#1}$}
\newcommand{\annulus}[1]{$A_{#1}$}

\newcommand{\spread}[1]{Sp(#1)}

\newcommand{\canonize}[1]{\textit{Canonize$(#1)$}}
\newcommand{\needspace}[1]{\textit{Space$(#1)$}}

\begin{document}	
	\title{Pole Dancing:\thanks{We here refer to pole dancing as a fitness and
competitive sport. The authors hope that many of our readers try this
activity themselves, and will in return introduce many pole dancers to
Graph Drawing, thereby alleviating the gender imbalance in both
communities.
The authors do not condone any pole activity used for sexual
exploitation or abuse of women or men.} 3D Morphs for Tree Drawings\thanks{E. A. was partially supported by F.R.S.-FNRS and SNF grant P2TIP2-168563 under the SNF Early PostDoc Mobility program. 
P.C. was supported by CONACyT, projects MINECO MTM2015-63791-R and Gen. Cat. 2017SGR1640.
P.B, A.D and V.D. were supported by NSERC.
F.F. was partially supported by MIUR Project ``MODE" under PRIN 20157EFM5C and by H2020-MSCA-RISE project 734922, ``CONNECT".
S.L. is Directeur de Recherches du F.R.S.-FNRS.
A.T. was partially supported by the project ``Algoritmi e sistemi di analisi visuale di reti complesse e di grandi dimensioni" - Ric. di Base 2018, Dip. Ingegneria, Univ. Perugia.}}

	\author{Elena Arseneva\inst{1} \and
		Prosenjit Bose\inst{2} \and
		Pilar Cano \inst{2,3} \and
		Anthony D'Angelo\inst{2} \and \\
		Vida Dujmovi\'c\inst{4} \and
		Fabrizio Frati\inst{5} \and
		Stefan Langerman\inst{6} \and
		Alessandra Tappini\inst{7}}
	
	\authorrunning{E. Arseneva et al.}
	\institute{
St. Petersburg State University (SPbU), Russia \email{ea.arseneva@gmail.com}
 \and 
		Carleton University, Ottawa, Canada\\ 
		\email{jit@scs.carleton.ca, anthonydangelo@cmail.carleton.ca} \and
		Universitat Polit\`ecnica de Catalunya, Barcelona, Spain
		\email{m.pilar.cano@upc.edu} \and
		University of Ottawa, Canada \email{vida@cs.mcgill.ca} \and
		Roma Tre University, Italy \email{frati@dia.uniroma3.it} \and
		Universit\'e libre de Bruxelles (ULB)
		\email{stefan.langerman@ulb.ac.be} \and
		Universit\`a degli Studi di Perugia, Italy \email{alessandra.tappini@studenti.unipg.it}}

	\maketitle

	\begin{abstract}
		We study the question whether a crossing-free 3D morph between two straight-line drawings of an $n$-vertex tree can be constructed consisting of a small number of linear morphing steps. We look both at the case in which the two given drawings are two-dimensional and at the one in which they are three-dimensional. In the former setting we prove that a crossing-free 3D morph always exists with $O(\log n)$ steps, while for the latter $\Theta(n)$ steps are always sufficient and sometimes necessary.
	\end{abstract}
	
	\section{Introduction}
	A {\em morph} between two drawings of the same graph is a continuous transformation from one drawing to the other. Thus, any  time instant of the morph defines a different drawing of the graph. Ideally, the morph should preserve the properties of the initial and final drawings throughout. As the most notable example, a morph between two planar graph drawings should guarantee that every intermediate drawing is also planar; if this happens, then the morph is called {\em planar}.
	
	Planar morphs have been studied for decades and find nowadays applications in animation, modeling, and computer graphics; see, e.g.,~\cite{fg-hmti-99,gs-gifm-pm}. A planar morph between any two topologically-equivalent\footnote[2]{Two planar drawings of a connected graph are \emph{topologically equivalent} if they define the same clockwise order of the edges around each vertex and the same outer face.} planar straight-line\footnote[3]{A \emph{straight-line drawing} $\Gamma$ of a graph $G$ maps vertices to points in a Euclidean space and edges to open straight-line segments between the images of their end-vertices. We denote by $\Gamma(v)$ (by $\Gamma(G')$) the image of a vertex $v$ (of a subgraph $G'$ of $G$, resp.).} drawings of the same planar graph always exists; this was proved for maximal planar graphs by Cairns~\cite{c-dplc-44} back in 1944, and then for all planar graphs by Thomassen~\cite{t-dpg-83} almost forty years later. Note that a planar morph between two planar graph drawings that are not topologically equivalent does not exist.
	
	It has lately been well investigated whether a planar morph between any two topologically-equivalent planar straight-line drawings of the same planar graph always exists such that the vertex trajectories have low complexity. This is usually formalized as follows. Let $\Gamma$ and $\Gamma'$ be two topologically-equivalent planar straight-line drawings of the same planar graph $G$. Then a morph $\mathcal M$ is a sequence $\langle \Gamma_1,\Gamma_2,\dots,\Gamma_k\rangle$ of planar straight-line drawings of $G$ such that $\Gamma_1=\Gamma$, $\Gamma_k=\Gamma'$, and $\langle \Gamma_i,\Gamma_{i+1}\rangle$ is a planar linear morph, for each $i=1,\dots,k-1$. A {\em linear morph} $\langle \Gamma_i,\Gamma_{i+1}\rangle$ is such that each vertex moves along a straight-line segment at uniform speed; that is, assuming that the morph happens between time $t=0$ and time $t=1$, the position of a vertex $v$ at any time $t\in[0,1]$ is $(1-t)\Gamma_i(v) + t\Gamma_{i+1}(v)$. The complexity of a morph $\mathcal M$ is then measured by the number of its {\em steps}, i.e., by the number of linear morphs it consists of.
	
	A recent sequence of papers~\cite{SODA-morph,Angelini-optimal-14,angelini2013morphing,Barrera-unidirectional} culminated in a proof~\cite{alamdari2017morph} that a planar morph between any two topologically-equivalent planar straight-line drawings of the same $n$-vertex planar graph can always be constructed consisting of $\Theta(n)$ steps. This bound is asymptotically optimal in the worst case, even for paths.
	
	The question we study in this paper is whether morphs with sub-linear complexity can be constructed if a third dimension is allowed to be used. That is: Given two topologically-equivalent planar straight-line drawings $\Gamma$ and $\Gamma'$ of the same $n$-vertex planar graph $G$ does a morph $\mathcal M=\langle \Gamma=\Gamma_1,\Gamma_2,\dots,\Gamma_k=\Gamma'\rangle$ exist such that: (i) for $i=1,\dots,k$, the drawing $\Gamma_i$ is a crossing-free straight-line 3D drawing of $G$, i.e., a straight-line drawing of $G$ in $\mathbb{R}^3$ such that no two edges cross; (ii) for $i=1,\dots,k-1$, the step $\langle \Gamma_i,\Gamma_{i+1}\rangle$ is a crossing-free linear morph, i.e., no two edges cross throughout the transformation; and (iii) $k=o(n)$? A morph $\mathcal M$ satisfying properties (i) and (ii) is a {\em crossing-free 3D morph}. 
	
	Our main result is a positive answer to the above question for trees. Namely, we prove that, for any two planar straight-line drawings $\Gamma$ and $\Gamma'$ of an $n$-vertex tree $T$, there is a crossing-free 3D morph with $O(\log n)$ steps between $\Gamma$ and $\Gamma'$. More precisely the number of steps in the morph is linear in the \emph{pathwidth} of $T$.  Notably, our morphing algorithm works even if $\Gamma$ and $\Gamma'$ are not topologically equivalent, hence the use of a third dimension overcomes another important limitation of planar two-dimensional morphs. Our algorithm morphs both $\Gamma$ and $\Gamma'$ to an intermediate suitably-defined {\em canonical 3D drawing}; in order to do that, a root-to-leaf path $H$ of $T$ is moved to a vertical line and then the subtrees of $T$ rooted at the children of the vertices in $H$ are moved around that vertical line, thus resembling a pole dance, from which the title of the paper comes.

	We also look at whether our result can be generalized to morphs of crossing-free straight-line 3D drawings of trees. That is, the drawings $\Gamma$ and $\Gamma'$ now live in $\mathbb{R}^3$, and the question is again whether a crossing-free 3D morph between $\Gamma$ and $\Gamma'$ exists with $o(n)$ steps. We prove that this is not the case: Two crossing-free straight-line 3D drawings of a path might require $\Omega(n)$ steps to be morphed one into the other. The matching upper bound
 can always be achieved: For any two crossing-free straight-line 3D drawings $\Gamma$ and $\Gamma'$ of the same $n$-vertex tree $T$ there is a crossing-free 3D morph between $\Gamma$ and $\Gamma'$ with $O(n)$ steps.

	The rest of the paper is organized as follows. 
	In Sect.~\ref{sec:3dtrees} we deal with crossing-free 3D morphs of 3D tree drawings. In Sect.~\ref{sec:morphing-paths} we show how to construct $2$-step crossing-free 3D morphs between planar straight-line drawings of a path. In Sect.~\ref{sec:morphing-trees} we present our main result about crossing-free 3D morphs of planar tree drawings. Finally, in Sect.~\ref{sec:conclusions} we conclude and present some open problems.
	
	Because of space limitations, some proofs are omitted or just sketched; they can be found in the full version of the paper.	

	\section{Morphs of 3D drawings of trees} \label{sec:3dtrees}
	
	In this section we give a tight $\Theta(n)$ bound on the number of steps 
	in a crossing-free 3D morph between two crossing-free straight-line 3D tree drawings.

	\begin{theorem} \label{th:upper-bound-3d}
		For any two crossing-free straight-line 3D drawings $\Gamma$, $\Gamma'$ of an $n$-vertex tree $T$, there exists a crossing-free 3D morph from $\Gamma$ to $\Gamma'$ that consists of $O(n)$ steps. 
	\end{theorem}

	\begin{proof}[sketch]
		The proof is by induction on $n$. The base case, in which $n=1$, is trivial. If $n>1$, then we remove a leaf $v$ and its incident edge $uv$ from $T$, $\Gamma$, and $\Gamma'$. This results in an $(n-1)$-vertex tree $T'$ and two drawings $\Delta$ and $\Delta'$ of it. By induction, there is a crossing-free 3D morph between $\Delta$ and $\Delta'$. We introduce $v$ in such a morph so that it is arbitrarily close to $u$ throughout the transformation; this significantly helps to avoid crossings in the morph. The number of steps is the one of the recursively constructed morph plus one initial step to bring $v$ close to $u$, plus two final steps to bring $v$ to its final position.\qed
	\end{proof}

	\begin{figure}[tb]
		\begin{minipage}{0.49\textwidth}
			\centering
			\includegraphics[scale = 0.6]{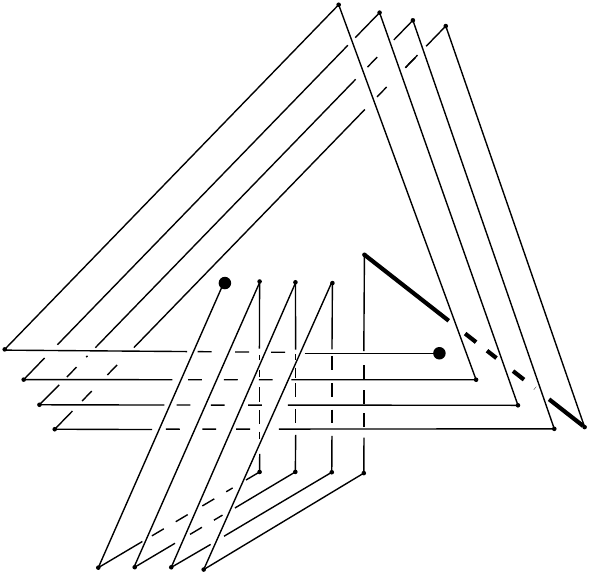}
			
			(a)
		\end{minipage}
		\begin{minipage}{0.49\textwidth}
			\centering
			\includegraphics[scale = 0.6,page=2]{TwoLinks-1.pdf}
			
			(b)
		\end{minipage}
		
		\caption{Illustration for the proof of Theorem~\ref{thm:lower-bound-3d}: (a) The drawing $\Gamma$ of $P$, with $n=26$; (b) the link obtained from $\Gamma$; the invisible edges are dashed. }
		\label{fig:lower-bound-3d}
	\end{figure}

	\begin{theorem}
		\label{thm:lower-bound-3d}
		There exist two crossing-free straight-line 3D drawings $\Gamma, \Gamma'$ of an $n$-vertex path $P$ such that any crossing-free 3D morph
		from $\Gamma$ to  $\Gamma'$ consists of $\Omega(n)$ steps.
	\end{theorem}

	Before proving Theorem~\ref{thm:lower-bound-3d}, we review some definitions and facts from knot theory; refer, e.g., to the book by Adams~\cite{adams2004knot}. A \emph{knot} is an embedding of a circle $S^{1}$ in $\mathbb{R}^3$. A {\em link} is a collection of knots which do not intersect, but which may be linked together. For links of two knots, the (absolute value of the) {\em linking number} is an invariant that classifies links with respect to ambient isotopies. Intuitively, the linking number is the number of times that each knot winds around the other. The linking number is known to be invariant with respect to different projections of the same link~\cite{adams2004knot}. Given a projection of the link, the linking number can be determined by orienting the two knots of the link, and for every crossing between the two knots in the projection adding $+1$ or $-1$ if rotating  the understrand respectively  clockwise or counterclockwise  lines it up with the overstrand (taking into account the direction).

	\begin{proof}[Theorem~\ref{thm:lower-bound-3d}]
		The drawing $\Gamma$ of $P$ is defined as follows. Embed the first $\lfloor n/2 \rfloor$ edges of $P$ in 3D as a spiral of monotonically decreasing height. Embed the rest of $P$ as a same type of spiral affinely transformed so that it goes around one of the sides of the former spiral. See Figure~\ref{fig:lower-bound-3d}a. The drawing $\Gamma'$ places the vertices of  $P$ in order along the unit parabola in the plane $y = 0$.

		Cut the edge joining the two spirals (the bold edge in Figure~\ref{fig:lower-bound-3d}a). Removing an edge makes morphing easier so any lower bound would still apply. Now close the two open curves using two \emph{invisible} edges to obtain a \emph{link} 
		of two knots; see Figure~\ref{fig:lower-bound-3d}b.
			It is easy to verify that the (absolute value of the) linking number of this link is $\Omega(n^2)$: indeed, determining it by the above procedure for the projection given by Figure~\ref{fig:lower-bound-3d} results in the linking number being equal to the number of crossings between the two links in this projection. In the drawing $\Gamma'$,  each of the two halves of $P$ (and their invisible edges) are separated by a plane and so their linking number is $0$. 
	
	In a valid linear morph, the edges of $P$ cannot cross each other, but they can cross invisible edges.
		However, during a linear morph between two straight-line 3D drawings of a graph $G$ any two non-adjacent edges of $G$ intersect $O(1)$ times.	
	Thus each invisible edge can only be crossed $O(n)$ times during a linear morph. A single crossing can only change the linking number by 1. Therefore the linking number can only decrease by $O(n)$ in a single linear morph. \qed
	\end{proof}

	\section{Morphing two planar drawings of a path in 3D}\label{sec:morphing-paths}
	
	In this section we show how to morph two planar straight-line drawings $\Gamma$ and $\Gamma'$ of an $n$-vertex path $P:=(v_0,\dots v_{n-1})$ into each other in two steps.

	The \emph{canonical 3D drawing} of $P$, denoted by $\canon{P}$, is the crossing-free straight-line 3D drawing of $P$ that maps each vertex $v_i$ to the point $(0,0,i) \in \mathbb{R}^3$, 
	as shown in Figure~\ref{fig:canonical-path}. We now prove the following.

	\begin{figure}[tb]
		\centering
		\includegraphics[scale=.6]{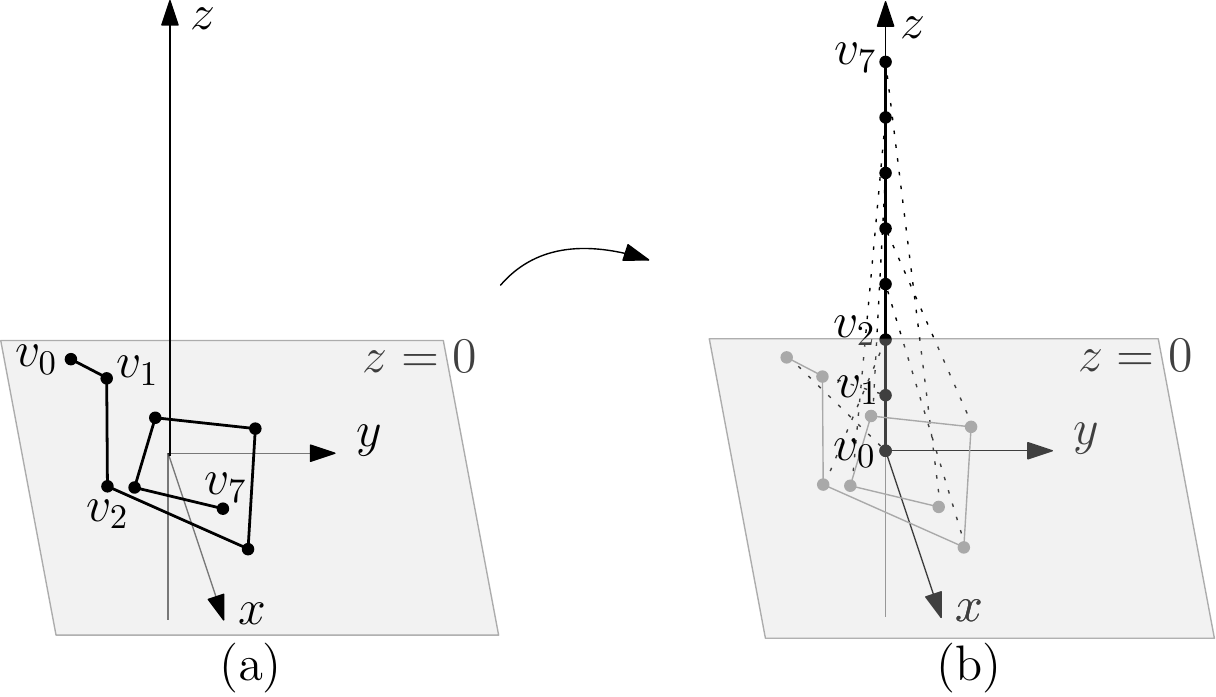}
		\caption{(a) A straight-line planar drawing $\Gamma$ of an $n$-vertex path $P$ and (b) a morph from $\Gamma$ to $\canon{P}$. The vertex trajectories are represented by dotted lines.} 
		\label{fig:canonical-path}
	\end{figure}

	\begin{theorem}
		\label{thm:path}
		For any two planar straight-line drawings $\Gamma$ and $\Gamma'$ of an $n$-vertex path $P$, there exists a crossing-free 3D morph $\mathcal{M}=\langle \Gamma,\canon{P},\Gamma' \rangle$ with 2 steps.
	\end{theorem}
	
	\begin{proof}
		
		It suffices to prove that the linear morph $\langle \Gamma,\canon{P}\rangle$ is crossing-free, since the morph $\langle \canon{P},\Gamma'\rangle$ is just the morph $\langle\Gamma',\canon{P}\rangle$ played backwards.

		Since $\langle \Gamma,\canon{P}\rangle$ is linear, the speed at which the vertices of $P$ move is  uniform (though it might be different for different vertices). Thus the speed at which their projections on the $z$-axis move is uniform as well.	
		Since $v_i$ moves uniformly from $(x_i,y_i,0)$ to $(0,0,i)$, at any time during the motion (except at the time $t=0$) we have 
		$z(v_0) < z(v_1) < \ldots < z(v_{n-1})$. 
		Therefore, in any intermediate drawing any edge $(v_i,v_{i+1})$ 
		is separated from any other edge by the horizontal plane through one of its end-points. 
		Hence no crossing happens during $\langle \Gamma,\canon{P}\rangle$.
 \qed
	\end{proof}

	\section{Morphing 
		two planar drawings of a tree in 3D}
	\label{sec:morphing-trees}
	
	Let $T$ be a tree with $n$ vertices, arbitrarily rooted at any vertex. 
In this section we show that any two planar straight-line drawings of $T$ can be morphed into one another by a crossing-free 3D morph with $O(\log n)$ steps (Theorem~\ref{thm:morph-trees}). 
	Similarly to Section~\ref{sec:morphing-paths}, we first define a canonical 3D drawing $\canon{T}$ of $T$  (see Section~\ref{sec:canonical-tree}), and then show how to construct a crossing-free 3D morph from any planar straight-line drawing of $T$ to 
	$\canon{T}$. 
	We describe the morphing procedure in Section~\ref{sec:morph-tree}; then in Section~\ref{subsec:space} we present a procedure \needspace{ } that carries out the computations required by the morphing procedure; finally, in Section~\ref{subsec:morph-analysis} we analyze the correctness and efficiency of both procedures.     
	
	Before proceeding, we introduce some necessary definitions and notation.  	
	By a {\em cone} we mean a straight circular cone induced by a ray rotated around a fixed vertical line (the {\em axis}) while keeping its origin fixed at a point (the {\em apex}) on this line. The {\em slope}  $\slope{C}$ of a cone $C$, is the slope of the generating ray as determined in the vertical plane  containing the ray.  
	By a {\em cylinder}
	we always mean a straight cylinder having a horizontal circle as a base. Such cones or cylinders are uniquely determined, up to translations, respectively by their apex and slope or by their height and radius.

	For a tree $T$, let $T(v)$ denote the subtree of $T$ rooted at its vertex $v$. Also let $|T|$ denote 
	the number of vertices in $T$. 	
	The {\em heavy-path decomposition}~\cite{sleator1981data} of a 
	tree $T$ is defined as follows. For each non-leaf vertex $v$ of $T$, let $w$ be the child of $v$ in $T$ such that $|T(w)|$ is maximum (in case of a tie, we choose the child arbitrarily). Then $(v,w)$ is a {\em heavy edge}; further, each child $z$ of $v$ different from $w$ is a \emph{light child} of $v$, and the edge $(v,w)$ is a {\em light edge}. Connected components of heavy edges form paths, called \emph{heavy paths}, which may have many incident light edges.
	Each path has a vertex, called the \emph{head}, that is the closest vertex to the root of $T$. 
	See Figure~\ref{fig:heavy-path-cactus} for an example.
	A \emph{path tree} of $T$ is a tree whose vertices correspond to 
	heavy paths in $T$. 
	The parent of a heavy path $P$ in the path tree is the heavy path that contains the parent of the head of $P$. 
	The root of the path tree is the heavy path containing the root of $T$. It is well-known~\cite{sleator1981data} that the height of the path tree is $O(\log{n})$.
	We denote by $H(T)$ the root of the path tree of $T$; let $v_0, \ldots, v_{k-1}$ be the ordered sequence of the vertices of $H(T)$, where $v_0$ is the root of $T$. For $i=0,\dots,k-1$, we let $v_i^0, \ldots, v_i^{t_i}$ be the light children of $v_i$ in any order.	
	Let $L(T) = u_0, u_1, \ldots, u_{l-1}$ be the sequence of the light children of $H(T)$ ordered so that: (i) any light child of a vertex $v_j$ precedes any light child of a vertex $v_i$, if $i < j$; and (ii) the light child $v_i^{j+1}$ of a vertex $v_i$ precedes the light child $v_i^j$ of $v_i$. When there is no ambiguity we refer to $H(T)$ and $L(T)$ simply as $H$ and $L$, respectively.

	\begin{figure}[tb]
		
			\centering
			\includegraphics[scale = 0.7]{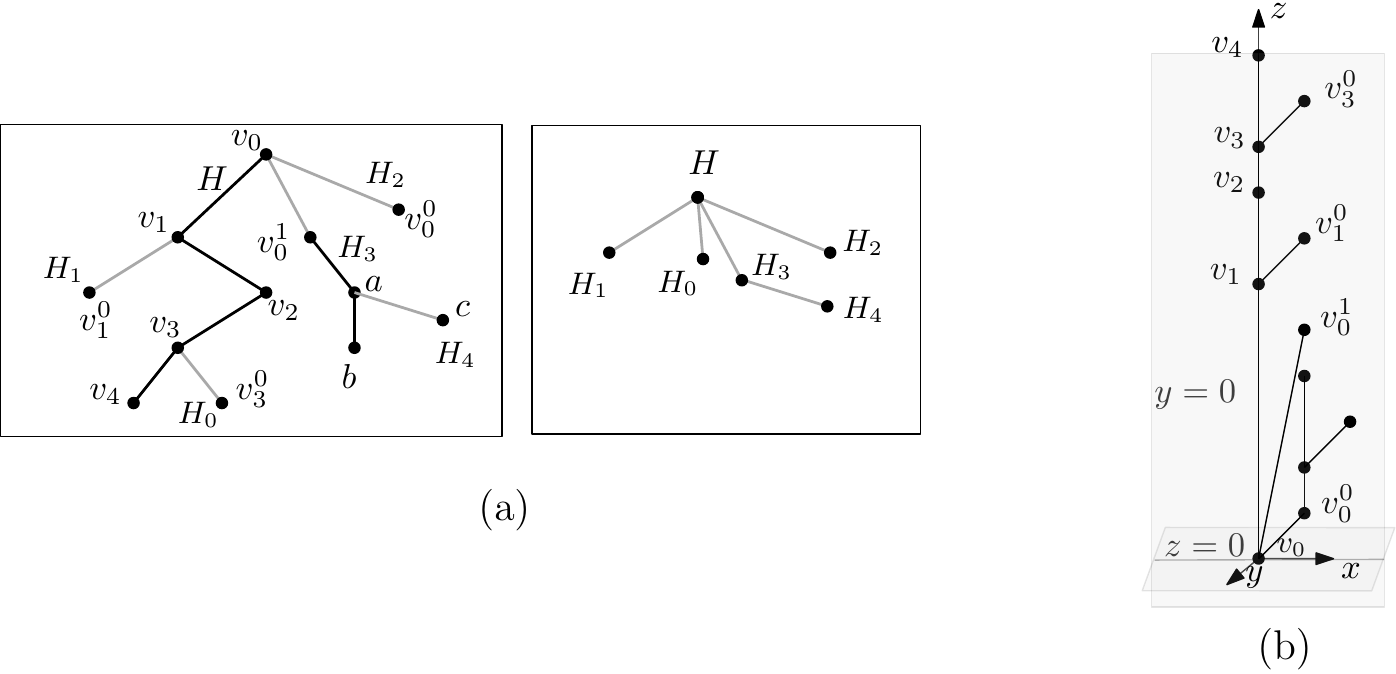}

		\caption{(a) A tree $T$; (left) its heavy  edges (bold lines) forming the heavy paths $H=H(T),H_0, \ldots, H_4$, and (right) the path tree of $T$; (b) $\canon{T}$ for the tree $T$ in (a). }
		\label{fig:heavy-path-cactus}
	\end{figure}

	\subsection{Canonical 3D drawing of a tree}
	\label{sec:canonical-tree}

We define the {\em canonical 3D drawing} $\canon{T}$ of a tree $T$ as a straight-line 3D drawing of $T$ that maps each vertex $v$ of $T$ to its {\em canonical position} $\canon{v}$ defined 
as follows (see  Figure~\ref{fig:heavy-path-cactus}b). Note that our canonical drawing is equivalent to the \emph{``standard''} straight-line upward drawing of a tree~\cite{therese,chan2018tree,crescenzi1992note}.

First, we set $\canon{v_0} = (0,0,0)$ for the root $v_0$ of $T$. 
Second, for each $i=1,\dots,k-1$, we set  $\canon{v_i} = (0,0, z_{i-1} + |T(v_{i-1})|-|T(v_i)|)$, where $z_{i-1}$ is the $z$-coordinate of $\canon{v_{i-1}}$. 
Third, for each $i=1,\dots,k-1$ and for each light child $v_i^j$ of $v_i$, we determine $\canon{v_i^j}$ as follows. If $j = 0$, we set $\canon{v_i^j} = (1,0,1+z_i)$, where $z_i$ is the $z$-coordinate of $\canon{v_i}$; otherwise, we set $\canon{v_i^j} = (1,0,z^{j-1}_i+|T(v_i^{j-1})|)$, where $z^{j-1}_i$ is the $z$-coordinate of $\canon{v^{j-1}_i}$. 
Finally, in order to determine the canonical positions of the vertices in $T(v_i^j) \setminus \{v_i^j\}$, we recursively construct the canonical 3D drawing $\canon{T(v_i^j)}$ of ${T(v_i^j)}$, and translate all the vertices by the same vector so that $v_i^j$ is sent to $\canon{v_i^j}$. 

\begin{Remark}
Notice that the canonical position $\canon{v}$ of any vertex $v$ of $T$ is $({\emph{dpt}}(v), 0, {\emph{dfs}}(v))$. Here ${\emph{dpt}}(v)$ is the depth, in the path tree of $T$, of the node that corresponds to the 
heavy path of $T$ that contains $v$; and 
${\emph{dfs}}(v)$ is the position of $v$ in a depth-first search on $T$ 
in which the children of any vertex are visited as follows: first visit the light children in reverse order with respect to $L$, and then visit the child incident to the heavy edge. 
\end{Remark}

The following lemma is a direct consequence of the construction of $\canon{T}$.

	\begin{lemma}
		\label{lemma:canon-tree}
		The canonical 3D drawing $\canon{T}$ of $T$ lies on a rectangular grid
in the plane $y=0$, where the grid has height $n$ and width equal to the  
height $h=O(\log{n})$ of the path tree of $T$. Moreover, $\canon{T}$ is on or above the line $z = x$. 
		\end{lemma}

\begin{Remark}
\label{rem:therese}
In the above definition of the canonical 3D drawing $\canon{T}$, instead of the heavy-path decomposition of $T$,
 we can use the decomposition based on the \emph{Strahler number} of $T$, see~\cite{therese} where the Strahler number is used under the name \emph{rooted pathwidth} of $T$.
 With this change, the width of $\canon{T}$ will be equal to the Strahler number of $T$, 
which is the instance-optimal width of an upward drawing of a tree~\cite{therese}. Moreover, since
the Strahler number is linear in the \emph{pathwidth} of $T$, so is the width of 
$\canon{T}$ defined this way. 
 This is clearly not worse, and, for some instances, much better than the width given by the heavy-path decomposition. 
\end{Remark}	
In the below description of the morph we use heavy paths, however we can use the paths given by Remark~\ref{rem:therese} instead, without any modification. 

	\subsection{The procedure {\em Canonize($\Gamma$)}}
	\label{sec:morph-tree}
	
Let $\Gamma = \Gamma({T})$ be a planar straight-line drawing of a tree $T$. 
Below we give a recursive procedure \canonize{\Gamma} that constructs a crossing-free 3D morph from 
$\Gamma$ to the canonical 3D drawing $\canon{T}$. We assume that 
$\Gamma$ is enclosed in a disk of diameter $1$ centered at $(0,0,0)$ in the plane $z=0$, and that the root $v_0$ of ${T}$ is placed at $(0,0,0)$ in $\Gamma$. 
This is not a loss of generality, up to a suitable modification of the reference system.

	\paragraph{Step 1 (set the pole).}
	
	The first step of the procedure \canonize{\Gamma} 
	aims to construct a 
	  linear morph $\langle \Gamma,\Gamma_1\rangle$, where $\Gamma_1$ is such that the heavy path $H=(v_0,\dots,v_{k-1})$ of $T$ lies on the vertical line through $\Gamma(v_0)$ and the subtrees of $T$ rooted at the
	light children of each vertex $v_i$ lie on the horizontal plane through~$v_i$.
	More precisely, the vertices of $T$ are placed in $\Gamma_1$ as follows. For $i=0,\dots,{k-1}$, place $v_i$ at the point $\canon{v_i}$. 
	Every vertex that belongs to a subtree rooted at a light child of $v_i$ is placed at a point such that  its trajectory in the morph 
defines the same vector as the trajectory of $v_i$.\footnote{Since the morph  $\langle \Gamma,\Gamma_1\rangle$ is linear, the trajectory of any vertex $v$ is simply the line segment connecting the positions of $v$ in $\Gamma$ and in $\Gamma_1$. To define a  vector, we orient the segment towards the position of $v$ in $\Gamma_1$.} 
	Below we refer to $\Gamma_1(H)$ as the \emph{pole}. The pole will remain still throughout the rest of the morph.

	\paragraph{Step 2 (lift).} The aim of the second step of the procedure \canonize{\Gamma} is to construct a linear morph $\langle \Gamma_1,\Gamma_2\rangle$, where $\Gamma_2$ is such that the drawings of any two subtrees $T(u_i)$ and $T(u_j)$ rooted at different light children $u_i$ and $u_j$ of vertices in $H$ are vertically and horizontally separated. The separation between $\Gamma_2(T(u_i))$ and $\Gamma_2(T(u_j))$ is set to be large enough so that the recursively computed morphs \canonize{\Gamma_2(T(u_i))} and \canonize{\Gamma_2(T(u_j))} do not interfere with each other. 
	
	We describe how to construct $\Gamma_2$. As anticipated, $\Gamma_2(v_i)=\Gamma_1(v_i)$, for each vertex $v_i$ in $H$. 
	In order to determine the placement of the vertices not in $H$ we use $l$ cones \innerCone{u_0},  \ldots,\innerCone{u_{l-1}} and $l$ cones \outerCone{u_0},  \ldots,\outerCone{u_{l-1}}, namely one cone \innerCone{u_t} and one cone \outerCone{u_t} per vertex $u_t$ in $L$.
	We also	use, for each $u_t$, a cylinder \needspace{\Gamma_2(T(u_t))} that bounds the volume used by
	\canonize{\Gamma_2(T(u_t))}.
	We defer the computation of these cones and cylinders to Section~\ref{subsec:space}, and for now assume that they
	are already available.
	For each $t=0,\dots,l-1$ and for each $j=0,\dots,t-1$, assume that $\Gamma_2(T(u_j))$ has been computed already -- this is indeed the case when $t=0$. Let $\mathcal{P}_t$ be the horizontal plane $z=|T|-1 + \sum_{j = 0}^{t-1}{h(u_j)}$, where $h(u_j)$ is the height of the cylinder \needspace{\Gamma_2(T(u_j))}. 
	The drawing $\Gamma_2$ maps the subtree $T(u_t)$ to the plane $\mathcal{P}_t$, just outside the cone \innerCone{u_t} and just inside the cone \outerCone{u_t}. 
	See Figure~\ref{fig:cones}.
	We proceed with the formal 
	definition of $\Gamma_2$.  
	Let $v$ be any vertex of $T(u_t)$ and let $(v_x, v_y, v_z)$ be the coordinates of $\Gamma_1(v)$. 
	Then $\Gamma_2(v)$ is the point  $(v_x\frac{r_t}{r}, v_y\frac{r_t}{r}, z_t)$. 
	Here $z_t$ is the height of the plane $\mathcal{P}_t$, 
	$r_t$ is the radius of the section of \innerCone{u_t} by the plane $\mathcal{P}_t$, and  $r$ is the distance from $\Gamma_1(v_i)$
	to its closest point of the drawing $\Gamma_1(T(u_t))$, where $v_i$ is the parent of $u_t$.
	See Figure~\ref{fig:cones}.
	Note that the latter closest point can be a point on an edge.

	\begin{figure}
		\centering
		\includegraphics[scale=0.9]{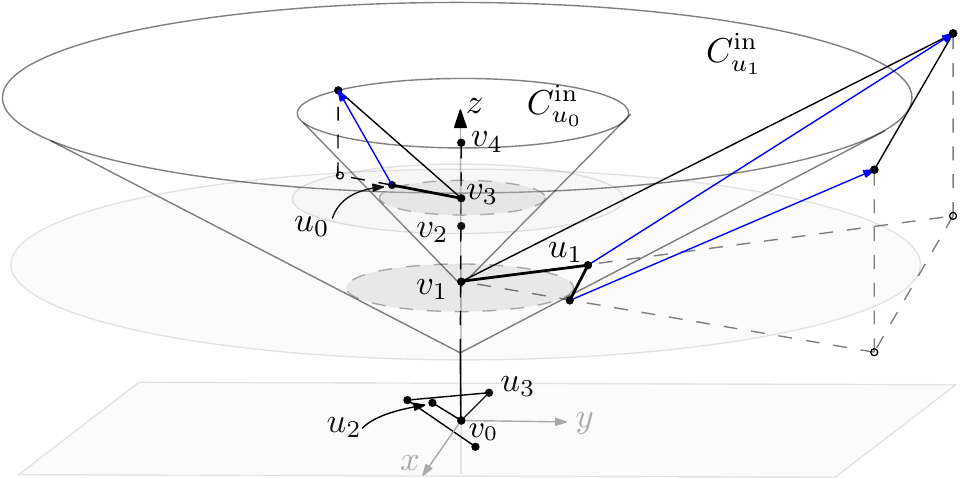}
		\caption{The vertices $v_0, v_1, v_2, v_3, v_4$ are in the heavy path $H$ of $T$. The lower gray disk has its center at $v_1$ and has radius equal to the distance from $\Gamma_1(v_i)$ to its closest point in $\Gamma_1(T(u_1))$. Blue arrows show the mapping of vertices in subtrees $T(u_0)$ and $T(u_1)$.} 
		\label{fig:cones}
	\end{figure}

	\paragraph{Step 3 (recurse).}
	
	For each $u_{t} \in L$, we make a recursive call  \canonize{\Gamma_2(T(u_t))}. The resulting morphs are combined into a unique morph $\langle \Gamma_2,\dots,\Gamma_3\rangle$, whose number of steps is equal to the maximum number of steps in any of the recursively computed morphs. Indeed, the first step of $\langle \Gamma_2,\dots,\Gamma_3\rangle$ consists of the first steps of all the recursively computed morphs that have at least one step; the second step of $\langle \Gamma_2,\dots,\Gamma_3\rangle$ consists of the second steps of all the $t$ recursively computed morphs that have at least two steps; and so on.

	\paragraph{Step 4 (rotate, rotate, rotate).} The next morph transforms $\Gamma_3$ into a drawing $\Gamma_4$ such that each vertex $u_t\in L$ is mapped to the intersection of the 
	cone \innerCone{u_t}, the planes $y=0$, $\mathcal{P}_t$, and the half-space $x>0$. 
	Note that going from $\Gamma_3$ to $\Gamma_4$ in one linear crossing-free 3D morph is not always possible. 
	Refer to Lemma~\ref{lemma:rotations} for the implementation of the morph from $\Gamma_3$ to $\Gamma_4$ in $O(1)$ steps. After Step 4 the whole drawing lies on the plane $y=0$.

	\paragraph{Step 5 (go down).}
	This step consists of a single linear morph $\langle \Gamma_4,\Gamma_5\rangle$, where $\Gamma_5$ is defined as follows. For every vertex $v_i$ in $H$, $\Gamma_5(v_i)=\Gamma_4(v_i)$; further, for every vertex $u_t \in L$, all the vertices of $T(u_t)$ have the same $x$- and $y$-coordinates in $\Gamma_5$ as in $\Gamma_4$, however their $z$-coordinate is decreased by the same amount so that $u_t$ lies on the horizontal plane through
 $\canon{u_t}$.

	\paragraph{Step 6 (go left).}
	The final part of our morphing procedure consists of a single linear morph $\langle \Gamma_5,\Gamma_6\rangle$, where $\Gamma_6$ is the canonical 3D drawing $\canon{T}$ of $T$. 
Note that this linear morph only moves the vertices 
horizontally.

	\subsection{The procedure {\em Space($\Gamma$)}}
	\label{subsec:space}
	
	In this section we give a procedure to compute the cylinders and the cones which are necessary for Steps 2 and 4 of the procedure \canonize{\Gamma}.

 	We fix a constant $c \in \mathbb{R}$ with $c>1$, which we consider global to the procedure \canonize{\Gamma} and its recursive calls; below we refer to $c$ as the \emph{global constant}. The global constant $c$ will help us to define the cones so that Step~4 of \canonize{\Gamma} can be realized with $O(1)$ linear morphs, see Lemma~\ref{lemma:rotations}.

The procedure \needspace{\Gamma} returns 	a 
	cylinder that encloses  all the intermediate drawings 
	of the morph determined by
	\canonize{\Gamma}.
	At the same time, 
	\needspace{\Gamma}

	determines the cones \innerCone{u_t} and \outerCone{u_t}
	for every vertex $u_t\in L$. 	
	
	We now describe \needspace{\Gamma}. 
	Let $\Gamma_1$ be the result of the application of Step 1 of \canonize{\Gamma}. Figure~\ref{fig:space-procedure} illustrates our description. 
	
	If $T$ is a path, i.e., $T = H$, return the cylinder of height $|T|-1$ and radius~$1$. In particular, if $T$ is a single vertex, return  the disk of radius 1. Otherwise,
		construct the cylinder 
		and the cones in the following fashion:
		\begin{itemize}
			\item Set the current cone $C$ to be an infinite cone of slope 1. The apex of $C$ is determined as follows: starting with the apex being at the highest point of the pole, slide $C$ vertically  downwards until it touches the drawing $\Gamma_1(T(u_0))$. That is, the apex of $C$ is at the lowest possible position on the  pole such that the whole drawing $\Gamma_1(T(u_0))$ is outside of $C$. See Figure~\ref{fig:space-procedure}a. 
			\item Set the current height $h$ to be $|T|-1$. 
			\item Iterate through the 
			light children of $H$ in the order as they appear in $L$. 
			For every $u_t$ in $L$: 
			\begin{itemize}
				\item Set \innerCone{u_t} to be the current cone $C$.
				
				\item Add the height of \needspace{\Gamma_2(T(u_t))} to the current height $h$.
				
				\item Let $C'$ be the cone with the same apex as $C$ and with a slope
				defined so that the drawing $\Gamma_1(T(u_t))$ is in-between $C$ and $C'$,
				and $C$ is \emph{well-separated} from $C'$ with the global constant $c$. That is, 
				$\slope{C'} = 
				\min{(\slope{C}/\spread{u_t, \Gamma_1}, \slope{C}/c)}$, 
				where 
				$\spread{u_t, \Gamma_1}$ is the spread of the drawing  $\Gamma_1(T(u_t))$ with respect to the parent $v_i$ of $u_t$ in $H$. Namely 
				$\spread{u_t, \Gamma_1}$ is 
				the ratio between the outer and the inner radius of
				the minimum annulus centered at $v_i$ and enclosing the drawing $\Gamma_1(T(u_t))$. 
				See Figure~\ref{fig:space-procedure}a.

				\item Let $\mathcal{S}_t$ be the cylinder  
				\needspace{\Gamma_2(T(u_t))} translated
				so that 
				the center of its lower base is at 
				the point 
				$\Gamma_2(u_t)$.

				\item Decrease $\slope{C'}$ so that $C'$ encloses the entire cylinder $\mathcal{S}_t$.
				
				\item Set \outerCone{u_t} to be the 
				cone
				$C'$.
				
				\item  If $u_t$ is not the last element of $L$ (i.e., $t < l-1$), then let $u_t = v_i^j$ 
				and define an auxiliary cone $\tilde{C}$ as follows. The apex of $\tilde{C}$ is at 
					$\Gamma_1({v_x})$ where $v_x$ is the parent of $u_{t+1}$; note that $v_x=v_i$ iff $j>0$.
					The slope of $\tilde{C}$ is   the maximum   slope that satisfies the following requirement:   
					(i) the slope of $\tilde{C}$ is at most   the slope of $C'$.  
					In addition, only for 
					the case when $v_x = v_i$, we require: 
					(ii) in the closed half space 
					$z \leq h$, the portion of
					$\tilde{C}$ encloses the portion of $C'$. See Figure~\ref{fig:space-procedure}b.
						Update the cone $C$ to be the lowest vertical translate of $\tilde{C}$ so that $\Gamma_1(T(u_{t+1}))$ is still outside the cone.
			
			\end{itemize}

			\item 
			Return the cylinder 
			of height $h$ (the current height), and radius equal to the radius of the section of the current cone $C$ 
			cut by the plane $z=h$.

		\end{itemize}

\subsection{Correctness of the morphing procedure  
}
\label{subsec:morph-analysis}

In this section, we analyze the correctness and the efficiency of the procedure \canonize{\Gamma} (see Theorem~\ref{thm:morph-trees}) and we give the details of Step~4 (see Lemma~\ref{lemma:rotations}).

		\begin{figure}[htb]
			\begin{minipage}{0.49\linewidth}
				\centering
				\includegraphics[scale=0.9]{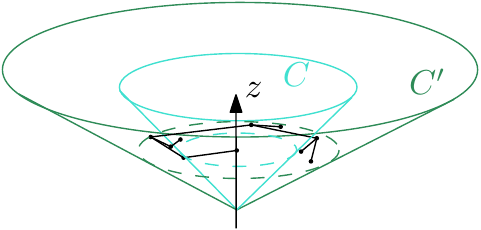}
				\\
				(a)
			\end{minipage}
			\hfill
			\begin{minipage}{0.49\linewidth}
				\centering
				\includegraphics[scale=0.9]{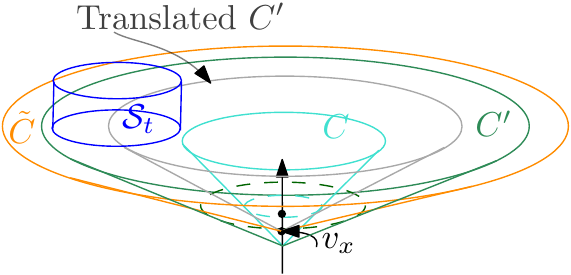}
				\\
				(b)
			\end{minipage}
			
			\caption{Illustration for \needspace{\Gamma}: (a) construction of $C$ and $C'$; (b) construction of $\tilde{C}$.}
			\label{fig:space-procedure}
		\end{figure}

		\begin{figure}
		\centering
		\includegraphics[scale=0.9]{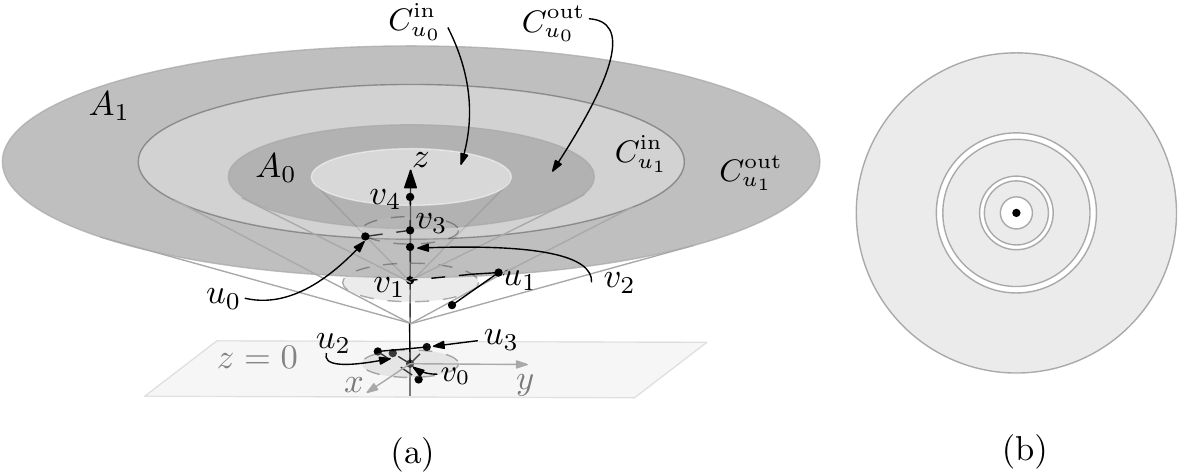}
		\caption{(a) Annuli for the subtrees rooted at $u_0$ and $u_1$; (b) top view of the annuli.}
		\label{fig:annuli}
	\end{figure}

\begin{lemma}
	\label{lemma:rotations}
	
	Step 4 of the procedure \canonize{\Gamma} can be realized as a
	crossing-free 3D  morph whose number of steps is bounded from above by a constant that depends on the global constant $c$.
\end{lemma}
\begin{proof}
Let \annulus{t} be the annulus formed by the section of \innerCone{u_t} and \outerCone{u_t}  cut by the plane $\mathcal{P}_t$. See Figure~\ref{fig:annuli}.
The morph performed in Step~4 consists of a sequence of linear morphs; in each of these morphs all the vertices of $T(u_t)$ are translated by the same vector. This is done so that $u_t$ stays in \annulus{t} during the whole Step 4. Thus, the trajectory of $u_t$ during Step~4 defines a polygon inscribed in \annulus{t}. Since the ratio between the outer and the inner radius of \annulus{t} is at least the global constant $c$, we can inscribe a regular $O(1)$-gon in \annulus{t}, and the trajectory of $u_t$ can be defined so that it follows this $O(1)$-gon plus at most one extra line segment.

We now prove that since each $u_t \in L$ stays in \annulus{t}, all the steps of the above morph are crossing-free. Recall that at any moment during the morph, the drawing of $T(u_t)$ is a translation of the canonical 3D drawing $\canon{T(u_t)}$. 
By Lemma~\ref{lemma:canon-tree}, the space below the line of slope 1 passing through $u_t$ in plane $y=0$ does not contain any point of $\canon{T(u_t)}$. Since the slope of \outerCone{u_t} is at most 1, the drawing of $T(u_t)$ is enclosed in \outerCone{u_t} as long as $u_t$ is in \annulus{t}. By conditions ($i$) and ($ii$) of \needspace{\Gamma}, the cone \innerCone{u_{t+1}} encloses  \outerCone{u_t} in the closed half-space above $\mathcal{P}_t$. Hence the edge connecting $u_{t+1}$ to the pole never touches \outerCone{u_t} above $\mathcal{P}_t$.
	 \qed

\end{proof}
	\begin{theorem}
	\label{thm:morph-trees}
		For any two plane straight-line drawings $\Gamma, \Gamma'$ of an $n$-vertex tree $T$, there exists a crossing-free 3D morph from $\Gamma$ to $\Gamma'$ with $O(\log n)$ steps.  
	\end{theorem}

\begin{proof}[sketch]
A 3D morph from $\Gamma$ to $\Gamma'$ can be constructed as the concatenation of \canonize{\Gamma} with the reverse of \canonize{\Gamma'}. Hence, it suffices to prove that \canonize{\Gamma} is a crossing-free 3D morph with $O(\log n)$ steps.

It is easy to see that Steps 1, 5, and 6 of \canonize{\Gamma} are crossing-free linear morphs. The proof that Step 2 is a crossing-free linear morph is more involved. In particular, for any two light children $u_s$ and $u_t$ with $s<t$ of the same vertex $v_i$ of $H$, the occurrence of a crossing between the edge $v_iu_s$ and an edge of $T(u_t)$ during Step 2 can be ruled out by arguing that the same two edges would also cross in $\Gamma_1$; this argument exploits the uniformity of the speed in a linear morph and that the horizontal component of the morph of Step~2 is a uniform scaling. Lemma~\ref{lemma:rotations} ensures that Step~4 is a crossing-free 3D morph with $O(1)$ steps. Thus, Steps 1, 2, 4, 5, and 6 require a total of $O(1)$ steps. Since the number of morphing steps of Step~3 of \canonize{\Gamma} is equal to the maximum number of steps of any recursively computed morph and since, by definition of heavy path, each tree $T(u_t)$ for which a recursive call \canonize{\Gamma_2(T(u_t))} is made has at most $n/2$ vertices, it follows that \canonize{\Gamma} requires $O(\log n)$ steps.
\qed
\end{proof}

	\section{Conclusions and Open Problems} \label{sec:conclusions}

	In this paper we studied crossing-free 3D morphs of tree drawings. We proved that, for any two planar straight-line drawings of the same $n$-vertex tree, there is a crossing-free 3D morph between them which consists of $O(\log n)$ steps.

	This result gives rise to two natural questions. First, is it possible to bring our logarithmic upper bound down to constant? In this paper we gave a positive answer to this question for paths. 
In fact our algorithm to morph planar straight-line tree drawings has a number of steps which is linear in the pathwidth of the tree (see Remark~\ref{rem:therese}), thus for example it is constant for caterpillars. Second, does a crossing-free 3D morph exist with $o(n)$ steps for any two planar straight-line drawings of the same $n$-vertex planar graph? The question is interesting to us even for subclasses of planar graphs, like outerplanar graphs and planar $3$-trees.  
		
	We also proved that any two crossing-free straight-line 3D drawings of an $n$-vertex tree can be morphed into each other in $O(n)$ steps; such a bound is asymptotically optimal in the worst case. An easy extension of our results to graphs containing cycles seems unlikely. Indeed, the existence of a deterministic algorithm to construct a crossing-free 3D morph with a polynomial number of steps between two crossing-free straight-line 3D drawings of a cycle would imply that the unknot recognition problem is polynomial-time solvable. The {\em unknot recognition} problem asks whether a given knot is equivalent to a circle in the plane under an ambient isotopy. This problem has been the subject of investigation for decades; it is known to be in NP~\cite{hlp-ccklp-99} and in co-NP~\cite{l-ecktn-16}, however determining whether it is in P has been an elusive goal so far.

\subsubsection*{Acknowledgments} 
We thank Therese Biedl for pointing out Remark~\ref{rem:therese}.  
\\ 
The research for this paper started during the Intensive Research Program in Discrete, Combinatorial and Computational Geometry, which took place in Barcelona, April-June 2018. We thank Vera Sacrist\'an and Rodrigo Silveira for a wonderful organization and all the participants for interesting discussions.
\\

	\bibliographystyle{plain}
	\bibliography{pole}
\end{document}